\def\be{\begin{equation}}
	\def\ee{\end{equation}}
\def\ba{\begin{array}}
	\def\ea{\end{array}}

\def\mathbi#1{\text{\em #1}}

\documentclass[prl,showpacs,twocolumn,amsmath]{revtex4}
\usepackage{amsmath}
\usepackage{amsfonts}
\usepackage{mathrsfs}
\usepackage{amssymb}
\usepackage{pifont}
\usepackage{epsfig,subfigure,dsfont,amsthm,amsbsy,mathrsfs,amscd}
\usepackage{epstopdf}
\usepackage{bbm}
\usepackage{color}
\def\qed{\leavevmode\unskip\penalty9999 \hbox{}\nobreak\hfil
	\quad\hbox{\leavevmode  \hbox to.77778em{%
			\hfil\vrule   \vbox to.675em%
			{\hrule width.6em\vfil\hrule}\vrule\hfil}}
	\par\vskip3pt}
\usepackage{leftidx}
\newtheorem{theorem}{Theorem}

\newtheorem{lemma}{Lemma}
\newtheorem{observation}{Observation}
\newtheorem{proposition}{Proposition}
\input amssym.def
\begin{document}
	\title{\large\bf Distribution relationship of quantum battery capacity}
	
	\author{Yiding Wang, Xiaofen Huang and Tinggui Zhang$^{\dag}$}
	\affiliation{ School of Mathematics and Statistics, Hainan Normal University, Haikou, 571158, China \\
		$^{\dag}$ Correspondence to 050003@hainnu.edu.cn}
	
	\bigskip
	\bigskip
	
	\begin{abstract}
    We investigate the distribution relationship of quantum battery capacity. First, we prove that for two-qubit $X$-states, the sum of the subsystem battery capacities does not exceed the total system's battery capacity, and we provide the conditions under which they are equal. We then define the difference between the total system's and subsystems' battery capacities as the residual battery capacity ($RBC$) and show that this can be divided into coherent and incoherent components. Furthermore, we observe that this capacity monogamy relation for quantum batteries extends to general $n$-qubit $X$ states and any $n$-qubit $X$ state's battery capacity distribution can be optimized to achieve capacity gain through an appropriate global unitary evolution. Specifically, for general three -qubit $X$ states, we derive stronger distributive relations for battery capacity. Quantum batteries are believed to hold significant potential for outperforming classical counterparts in the future. Our findings contribute to the development and enhancement of quantum battery theory.
	\end{abstract}
	
	 \maketitle
	
	\section{I. Introduction}
	With advancements in quantum thermodynamics, some novel quantum devices have gradually been proposed in recent years. Quantum batteries, as the devices can store and release energy in a proper manner \cite{phhs,cmov,akmc,mfha,omk}, have been widely investigated within the field of quantum science and technology.
	
	The quantum mechanical prototype of a battery was introduced by R. Alicki and M. Fannes \cite{ramf}, who also presented the concept of ergotropy, which is defined as the maximum amount of energy that can be extracted from a quantum system through unitary operations. Due to the unique quantum features utilized for energy storage and release, quantum batteries have the potential to outperform classical counterparts, offering superior work extraction \cite{akmc,jmmf,hlss,srsv,gfld}, enhanced charging power \cite{fcfa,yyzt,drgma,ssmp,jygd,gzyc,fclm,jygu,crdr,fmvc,tkkl}, and increased capacity \cite{lgcc,yyas,strs}. Consequently, quantum batteries have been extensively studied both theoretically and experimentally, inspiring a range of research efforts, including the development of charging models \cite{gfjf,dfmc,gmfp,gmam,jcrm,fmaj} and multipartite quantum batteries \cite{tplj,drgm,kxhj,szft,fqdy,dws}. Additionally, quantum batteries leverage quantum correlations, such as entanglement, to improve energy storage processes , addressing limitations of traditional batteries \cite{akmc,dfmc,rsmp}. Significant progress has been made in recent research on quantum batteries \cite{hyyh,fmyf,bapmp,mlsx,cadms}. Yang et al. introduced the invariant subspace method to effectively represent the quantum dynamics of the Tavis-Cummings battery \cite{hyyh}. The authors in \cite{fmyf} investigated the charging and self-discharging performance of quantum batteries and found that charging energy is positively correlated with coherence and entanglement while self-discharging energy is negatively correlated with coherence. Non-reciprocity, arising from the breaking of time-reversal symmetry, has become an important tool in various quantum technology applications. In \cite{bapmp}, Ahmadi et al. explored the potential of non-reciprocity in quantum batteries and demonstrated that non-reciprocity can improve charging efficiency and enhance the energy storage of batteries under certain optimal conditions. 
	
	The capacity or work extraction is believed to be a crucial quantitative indicator of the quantum battery quality \cite{yyas,strs}. The authors in \cite{yyas} introduced a new definition of quantum battery capacity, which can be directly linked with the battery state entropy and the coherence and entanglement measures. Wang et al. \cite{ykwlz} investigated the dynamics of quantum battery capacity for Bell-diagonal states under Markovian channels on the subsystem and observed that capacity increases for certain Bell-diagonal states under the amplitude damping channel. Ali et al. \cite{aasak} derived analytical expressions for the maximal extractable work, ergotropy, and the capacity of finite spin quantum batteries. In \cite{tgzh}, the authors showed that for bipartite systems, the battery capacity with respect to one subsystem can be improved through local-projective measurements on another subsystem.  Inspired by previous studies on the distribution of quantum correlations  \cite{vcjk,mkaw,tjof,thgaf,ycohf,ykbyf}, in this work, we address a fundamental and important question regarding battery capacity: to what extent the battery capacity of a subsystem limits the capacity of other subsystems. $X$ states, including the well-known Bell state, Werner state and Greenberger-Horne-Zeilinger (GHZ) state, are very important quantum states and has a wide range of applications such as quantum teleportation \cite{mccyl,zjxjh,jlmsk,lmjr}, quantum super dense coding \cite{mhmp1}, and quantum communication \cite{mhphrh}. The research on these states not only deepens our understanding of quantum correlations such as entanglement and nonlocality, but also promotes the development of quantum battery technology. We find that a monogamy relation exists in battery capacity for the general $n$-qubit $X$-states for the first time. Specifically, for the three qubits $X$-states, we establish stronger distributive relations for battery capacity.
	
	The remainder of this paper is organized as follows: In Section 2, we present the distribution relation of battery capacity for any two-qubit $X$ state [Theorem 1\,] and define the concept of residual battery capacity (RBC). In addition, we find that there can always be a unitary evolution to optimize the capacity distribution to achieve battery capacity gain [Theorem 2\,]. In Section 3, we introduce the monogamy relation of capacity for the general $n$-qubit $X$ state [Theorem 3\,] and achieve capacity gain by a unitary evolution [Theorem 4\,]. In particular, for any three-qubit $X$ state, we prove stronger monogamy relations in battery capacity [Theorem 5\,]. Finally, we summarize and discuss our conclusions in the last section.
	
    \section{II. The distribution relation of battery capacity for two-qubit X state}
		In \cite{yyas}, the authors introduced a novel definition of quantum battery capacity, which remains constant during any unitary evolution. This definition is given by
	\begin{equation}\label{e2}
		\begin{split}
			\mathcal{C}(\rho;H)&=\sum_{i=0}^{d-1}\epsilon_i(\lambda_i-\lambda_{d-1-i})\\
			&=\sum_{i=0}^{d-1}\lambda_i(\epsilon_i-\epsilon_{d-1-i}),
		\end{split}
	\end{equation}
	where $\{\lambda_i\}$ and $\{\epsilon_i\}$ represent the energy levels of $\rho$ and the Hamiltonian $H=\sum_{i=0}^{d-1}\epsilon_i|\varepsilon\rangle\langle\varepsilon|$, respectively, arranged in descending order without loss of generality, i.e., $\lambda_0\geqslant\lambda_1\geqslant...\geqslant\lambda_{d-1}$ and $\epsilon_0\geqslant\epsilon_1\geqslant...\geqslant\epsilon_{d-1}$. 
	
	It is worth noting that the authors demonstrated in \cite{yyas} that the battery capacity $\mathcal{C}$ is a Schur-convex functional for $\rho$.
	\begin{proposition}
		If a state $\rho$ is majorized by $\varrho$ $(\rho\prec\varrho)$, then we have $\mathcal{C}(\rho;H)\leq\mathcal{C}(\varrho;H)$.
	\end{proposition}
	Inspired by the entanglement monogamy relation \cite{vcjk}, a natural and interesting idea is the monogamy relation of battery capacity. Let us first consider a simple scenario: the two qubits $X$ state.
	
	Given the two-qubit $X$ state $\rho_{AB}$ in the computational basis,
	$$
	\rho_{AB}=\left(\begin{array}{cccc}
		\rho_{11} & 0 & 0 & \rho_{14} \\
		0 & \rho_{22} & \rho_{23} & 0 \\
		0 & \rho_{32} & \rho_{33} & 0 \\
		\rho_{41} & 0 & 0 & \rho_{44}
	\end{array}\right),
	$$
	where $\sum_{i=1}^{4}\rho_{ii}=1$ and $\rho_{ii}\geq0$. The eigenvalues of $\rho_{AB}$ are
	\begin{equation*}
		\begin{split}
			&\lambda_0=\frac{1}{2}(\rho_{11}+\rho_{44}-\sqrt{(\rho_{11}-\rho_{44})^2+4|\rho_{14}|^2}),\\
			&\lambda_1=\frac{1}{2}(\rho_{11}+\rho_{44}+\sqrt{(\rho_{11}-\rho_{44})^2+4|\rho_{14}|^2}),\\
			&\lambda_2=\frac{1}{2}(\rho_{22}+\rho_{33}-\sqrt{(\rho_{22}-\rho_{33})^2+4|\rho_{23}|^2}),\\
			&\lambda_3=\frac{1}{2}(\rho_{22}+\rho_{33}+\sqrt{(\rho_{22}-\rho_{33})^2+4|\rho_{23}|^2}).
		\end{split}
	\end{equation*}
	Consider the following Hamiltonian of the whole system:
	\begin{small}
		\begin{equation}\label{e3}
			H_{AB}=\epsilon^A\sigma_3\otimes I_2+\epsilon^BI_2\otimes\sigma_3+\gamma\sigma_1\otimes\sigma_1,
		\end{equation}
	\end{small}
where $\sigma_i$ $(i=1,2,3)$ are the standard Pauli matrices, $\gamma\geq0$ is the interaction parameter, and $\epsilon^A\geqslant\epsilon^B\geqslant0$. $\gamma=0$ corresponds to the total interaction-free global Hamiltonian. The reduced density matrices of $\rho_{AB}$ with respect to subsystem $A$ and $B$ are
	$$
	\rho_{A}=\left(\begin{array}{cc}
		\rho_{11}+\rho_{22} & 0 \\
		0 & \rho_{33}+\rho_{44}
	\end{array}\right)
	$$
	and
	$$
	\rho_{B}=\left(\begin{array}{cc}
		\rho_{11}+\rho_{33} & 0 \\
		0 & \rho_{22}+\rho_{44}
	\end{array}\right),
	$$
	respectively. Therefore, one can compute $\mathcal{C}(\rho_{AB};H_{AB})$, $\mathcal{C}(\rho_{A};H_{A})$ and $\mathcal{C}(\rho_{B};H_{B})$ according to Eq. (\ref{e2}) and (\ref{e3}). Before presenting the main results, we first note the following facts.
	
	\begin{lemma}
	For an $n$-dimensional positive semidefinite matrix $P$, let its diagonal elements be denoted as $d_1, d_2..., d_n$, and $\lambda_1, \lambda_2,...,\lambda_n$ are denoted as its eigenvalues. Then we have 
	\begin{equation}\label{e4}
	(d_1, d_2,..., d_n)\prec(\lambda_1, \lambda_2,..., \lambda_n).
	\end{equation}
	\end{lemma}
	\begin{proof}
	Since $P$ is a positive semidefinite matrix, it has spectral decomposition
	\begin{equation}\label{e5}
	P=U^\dagger diag(\lambda_1, \lambda_2,..., \lambda_n)U,
	\end{equation}
	where $U=(u_{ij})_{n\times n}$ is a unitary matrix. According to (\ref{e5}), it is not difficult to see that
	\begin{equation}\label{e6}
	d_i=\sum_{j=1}^{n}\lambda_j|u_{ji}|^2, 1\leq i\leq n.
	\end{equation}
	Setting $Q=(|u_{ij}|^2)_{n\times n}^T$, combined with the properties of unitary matrices, it can be concluded that $Q$ is a doubly stochastic matrix (A non-negative square matrix where the sum of elements in each row and column is equal to 1). And based on Eq. (\ref{e6}), we have
	\begin{equation}\label{e7}
	(d_1, d_2,..., d_n)^T=Q(\lambda_1, \lambda_2,..., \lambda_n)^T.
	\end{equation}
	Due to the Hardy-Littlewood-P\'olya theorem \cite{hlp}, one has
	\begin{equation*}
		(d_1, d_2,..., d_n)\prec(\lambda_1, \lambda_2,..., \lambda_n).
	\end{equation*}
	\end{proof}

    In fact, Lemma 1 naturally implies a lower bound for battery capacity,
    	\begin{equation}\label{e8}
    		\mathcal{C}(\rho;H)\geq\mathcal{C}(\tau;H)\equiv L_D,
    	\end{equation}
    	where $\tau$ is the decoherent state of $\rho$.
    It is worth noting that this lower bound holds for any battery state. Therefore, for some high-dimensional states with eigenvalues that are difficult to compute, this lower bound provides a relatively fast and simple calculation method, as it only requires consideration of the diagonal elements of the state.
    
    \begin{lemma}
    For a 2-qubit incoherent state $\rho_\text{ic}$, its diagonal elements are denoted as $\rho_{11},... \rho_{44}$, then we have
    \begin{equation}\label{e9}
    \mathcal{C}(\rho_\text{ic}^A;H_A)+\mathcal{C}(\rho_\text{ic}^B;H_B)\leq\mathcal{C}(\rho_\text{ic};H_{AB}).
    \end{equation}
    \end{lemma}
    \begin{proof}
    We set the descending order of $\rho_{11},...,\rho_{44}$ as $\alpha_1\geq\alpha_2\geq\alpha_3\geq\alpha_4$. In fact, it can be calculate that the eigenvalues of the Hamiltonian $H_{AB}$ given by Eq. (\ref{e3}) are $$\pm\sqrt{(\epsilon^A+\epsilon^B)^2+\gamma^2},\, \pm\sqrt{(\epsilon^A-\epsilon^B)^2+\gamma^2}.$$
    
    So one has
     \begin{equation*}
    	\begin{split}
    		\mathcal{C}(\rho_\text{ic};H_{AB})&\geq2(\alpha_1-\alpha_4)(\epsilon^A+\epsilon^B)+2(\alpha_2-\alpha_3)(\epsilon^A-\epsilon^B)\\
    		&=2(\alpha_1+\alpha_2-\alpha_3-\alpha_4)\epsilon^A\\
    		&+2(\alpha_1+\alpha_3-\alpha_2-\alpha_4)\epsilon^B\\
    		&\geq2|\rho_{11}+\rho_{22}-\rho_{33}-\rho_{44}|\epsilon^A\\
    		&+2|\rho_{11}+\rho_{33}-\rho_{22}-\rho_{44}|\epsilon^B\\
    		&=\mathcal{C}(\rho_\text{ic}^A;H_A)+\mathcal{C}(\rho_\text{ic}^B;H_B).
    	\end{split}
    \end{equation*}
    \end{proof}
    An interesting question is under what conditions we have
    \begin{equation}\label{e10}
    	\mathcal{C}(\rho_\text{ic}^A;H_A)+\mathcal{C}(\rho_\text{ic}^B;H_B)=\mathcal{C}(\rho_\text{ic};H_{AB}).
    \end{equation}
   According to the proof in Lemma 2, the interaction parameter $\gamma=0$ first, and then one can conclude that Eq. (\ref{e10}) is equivalent to:
    \begin{equation*}
    \begin{split}
    &(i)\,\alpha_1+\alpha_2-\alpha_3-\alpha_4=|\rho_{11}+\rho_{22}-\rho_{33}-\rho_{44}|,\\
    &(ii)\,\alpha_1+\alpha_3-\alpha_2-\alpha_4=|\rho_{11}+\rho_{33}-\rho_{22}-\rho_{44}|.
    \end{split}
    \end{equation*}
    It is not difficult to find 4 orders of diagonal elements that satisfy the above conditions:
    \begin{equation}\label{e11}
    \begin{split}
    &\rho_{11}\geq\rho_{22}\geq\rho_{33}\geq\rho_{44},\\
    &\rho_{22}\geq\rho_{11}\geq\rho_{44}\geq\rho_{33},\\
    &\rho_{33}\geq\rho_{44}\geq\rho_{11}\geq\rho_{22},\\
    &\rho_{44}\geq\rho_{33}\geq\rho_{22}\geq\rho_{11}.
    \end{split}
    \end{equation}
	We now present the main results of this section.
	\begin{theorem}
	For the two-qubit $X$-state, the following monogamy relation holds:
	\begin{equation}\label{e12}
	\mathcal{C}(\rho_A;H_A)+\mathcal{C}(\rho_B;H_B)\leq\mathcal{C}(\rho_{AB};H_{AB}).
	\end{equation}
	When $\rho_{AB}$ is an incoherent state satisfied its diagonal elements satisfy one of the conditions in (\ref{e11}), and $H_{AB}$ is the interaction-free global Hamiltonian, the equality holds.
	\end{theorem}
	\begin{proof}
	Consider the two -qubit $X$ state under the computational basis,
	$$
	\rho_{AB}=\left(\begin{array}{cccc}
		\rho_{11} & 0 & 0 & \rho_{14} \\
		0 & \rho_{22} & \rho_{23} & 0 \\
		0 & \rho_{32} & \rho_{33} & 0 \\
		\rho_{41} & 0 & 0 & \rho_{44}
	\end{array}\right),
	$$
	and its eigenvalues are denoted as $\lambda_i\,(i=0,...,3)$. According to Lemma 1, one can obtain that
	\begin{equation*}
		(\rho_{11}, \rho_{22}, \rho_{33}, \rho_{44})\prec(\lambda_0, \lambda_1, \lambda_2, \lambda_3).
	\end{equation*}
	Let $\rho_d=diag\,(\rho_{11}, \rho_{22}, \rho_{33}, \rho_{44})$ be the decoherent state of $\rho_{AB}$. Thus, we have
	\begin{small}
	\begin{equation*}
		\begin{split}
			\mathcal{C}(\rho_{AB};H_{AB})&\geq\mathcal{C}(\rho_d;H_{AB})\\
			&\geq\mathcal{C}(\rho_A;H_A)+\mathcal{C}(\rho_B;H_B).
		\end{split}
	\end{equation*}
	\end{small}
	The first inequality is due to Proposition 1, noting that $\rho_{11}, \rho_{22}, \rho_{33}$ and $\rho_{44}$ are eigenvalues of $\rho_d$. In particular, if $\rho_{AB}$ is an incoherent state such that its diagonal elements satisfy one of the conditions (\ref{e11}), and the interaction parameter $\gamma=0$, then the two inequalities above become two equalities.
	\end{proof}
	
	From Theorem 1, we can observe that for general two qubits $X$ coherent state $\rho_{AB}$,
	\begin{equation}\label{e13}
	\mathcal{C}(\rho_{AB};H_{AB})-\mathcal{C}(\rho_{A};H_{A})-\mathcal{C}(\rho_{B};H_{B})>0.
	\end{equation} 
	We can define this difference as the residual battery capacity ($RBC$) $\bigtriangleup\mathcal{C}$. In fact, we can divide $RBC$ into incoherent part $RBC_{ic}$ and a coherent part $RBC_c$, where $RBC_{ic}$ is the difference between the capacity of the decoherent state $\rho_d$ and the sum of the capacity of the reduced states, and $RBC_c$ is the capacity difference between the battery state $\rho$ and the decoherent state $\rho_d$. The relationship between them is shown in Figure 1. The reason is not difficult to understand. The reduced matrices of the $X$ state inevitably lose some incoherent information and all coherent information. In other words, the quantum correlation between subsystems $A$ and $B$ is ignored when using the reduced density matrix.
    \begin{figure}[htbp]
		\centering
		\includegraphics[width=0.5\textwidth]{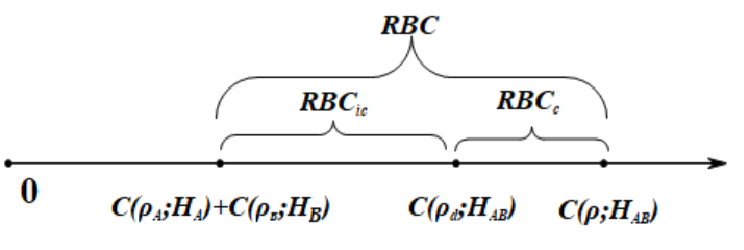}
		\vspace{-1em} \caption{The battery capacity distribution of two-qubit $X$ state.} \label{Fig.1}
	\end{figure}
	
	After obtaining the distributive relationship of battery capacity, a natural idea is to increase the capacity of subsystems without reducing the whole system capacity. Therefore, we can consider the unitary evolution of the battery state.
	
	Since unitary evolution does not change the eigenvalues of the state, the capacity of the entire battery system remains unchanged during this process according to Eq. (\ref{e2}). But we can use a special kind of unitary evolution to convert $RBC_{ic}$ into the battery capacity of the subsystem, so as to achieve the battery capacity gain, i.e.
	\begin{small}
	\begin{equation*}
		\mathcal{C}(\Tilde{\rho}_A;H_A)+\mathcal{C}(\Tilde{\rho}_B;H_B)-\mathcal{C}(\rho_A;H_A)-\mathcal{C}(\rho_B;H_B)>0,
	\end{equation*}
	\end{small}
	where $\Tilde{\rho}_A$ and $\Tilde{\rho}_B$ are reduced states of $\Tilde{\rho}=U\rho\,U^\dagger$.
	\begin{theorem}
	For an arbitrary two qubit $X$ state $\rho_{AB}$, there always exists a unitary evolution $U$ such that the subsystems can achieve battery capacity gain.
	\end{theorem}
	\begin{proof}
	Consider the two-qubit $X$ state under the computational basis,
	$$
	\rho_{AB}=\left(\begin{array}{cccc}
		\rho_{11} & 0 & 0 & \rho_{14} \\
		0 & \rho_{22} & \rho_{23} & 0 \\
		0 & \rho_{32} & \rho_{33} & 0 \\
		\rho_{41} & 0 & 0 & \rho_{44}
	\end{array}\right),
	$$
	and we assume that the order of diagonal elements are $\rho_{11}\geq\rho_{22}\geq\rho_{44}>\rho_{33}$ without loss of generality. Then we can consider the unitary evolution
	$$
	U=\left(\begin{array}{cccc}
		1 & 0 & 0 & 0 \\
		0 & 1 & 0 & 0 \\
		0 & 0 & 0 & 1 \\
		0 & 0 & 1 & 0
	\end{array}\right)
	$$
	such that 
	$$
	\Tilde{\rho}_{AB}=\left(\begin{array}{cccc}
		\rho_{11} & 0 & \rho_{14} & 0 \\
		0 & \rho_{22} & 0 & \rho_{23} \\
		\rho_{41} & 0 & \rho_{44} & 0 \\
		0 & \rho_{32} & 0 & \rho_{33}
	\end{array}\right).
	$$
	So the reduced density matrices of $\Tilde{\rho}_{AB}$ with respect to subsystem $A$ and $B$ are
	$$
	\Tilde{\rho}_{A}=\left(\begin{array}{cc}
		\rho_{11}+\rho_{22} & \rho_{14}+\rho_{23} \\
		\rho_{32}+\rho_{41} & \rho_{33}+\rho_{44}
	\end{array}\right)
	$$
	and
	$$
	\Tilde{\rho}_{B}=\left(\begin{array}{cc}
		\rho_{11}+\rho_{44} & 0 \\
		0 & \rho_{22}+\rho_{33}
	\end{array}\right),
	$$
	respectively. According to Lemma 1 and Proposition 1, we have $\mathcal{C}(\Tilde{\rho}_A;H_A)\geq\mathcal{C}(\rho_A;H_A)$. And $\mathcal{C}(\Tilde{\rho}_B;H_B)\geq\mathcal{C}(\rho_B;H_B)$ is due to the fact that
	\begin{equation*}
	\rho_{11}+\rho_{44}-\rho_{22}-\rho_{33}>|\rho_{11}+\rho_{33}-\rho_{22}-\rho_{44}|.
	\end{equation*}
	\end{proof}
	In fact, given a two qubit $X$ battery state $\rho_{AB}$, we can obtain its capacity distribution according to our theory. Then we can use a global unitary evolution to optimize its battery capacity distribution, because there is always a unitary matrix $U$ (possibly the product of a series of unitary matrices) such that the diagonal elements of $U\rho_{AB}U^\dagger$ satisfy one of the conditions in (\ref{e11}). We take the proof process of Theorem 2 as an example. When interaction parameter $\gamma=0$, we transfer part of $RBC_c$ to the capacity of subsystem $A$, and transfer all $RBC_{ic}$ to the capacity of subsystem $B$ through a unitary evolution. Therefore, the battery capacity gain of subsystem is realized. However, when $\gamma>0$, we can only transfer part of $RBC_c$ to the capacity of subsystem $A$, and part of $RBC_{ic}$ to the subsystem $B$, because the interaction between the two subsystems inevitably dissipates part of the battery capacity, the reduced state of sub-system only obtains part of the coherence information. The process schematic diagram is shown in Figure 2.
	
	\begin{figure}[htbp]
		\centering
		\includegraphics[width=0.45\textwidth]{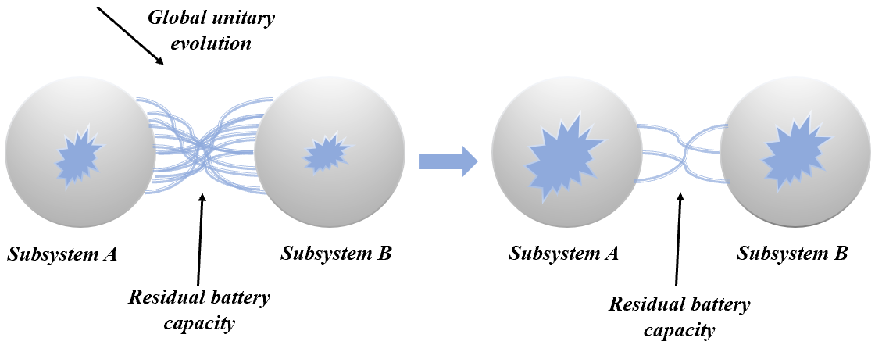}
		\vspace{-1em} \caption{The process of optimizing battery capacity distribution using global unitary evolution as shown in figure.} \label{Fig.2}
	\end{figure}
	
	\mathbi{Example 1}. Let us consider the quantum states to be 2-qubit Bell-diagonal ones given by \cite{tgzh},
	\begin{equation*}
	\varrho=\frac{1}{4}(I_2\otimes I_2+\sum_{i=1}^{3}a_i\sigma_i\otimes\sigma_i),
	\end{equation*} 
	where $a_i\,(i=1,2,3)$ are real constants such that $\varrho$ is a well-defined density matrix.
	
	The eigenvalues of $\varrho$ are
	\begin{small}
	\begin{equation*}
		\begin{split}
			&\lambda_0=\frac{1-a_1-a_2-a_3}{4}, \lambda_1=\frac{1-a_1+a_2+a_3}{4},\\
			&\lambda_2=\frac{1+a_1-a_2+a_3}{4}, \lambda_3=\frac{1+a_1+a_2-a_3}{4}.
		\end{split}
	\end{equation*}
	\end{small}
	Due to symmetry, we assume that $|a_1|\geq|a_2|\geq|a_3|$. Therefore, according to Eq. (\ref{e2}) and (\ref{e3}), we have
	\begin{small}
	\begin{equation*}
	\begin{split}
	\mathcal{C}(\varrho;H_{AB})&=(|a_1|+|a_2|)\sqrt{(\epsilon^A+\epsilon^B)^2+\gamma^2}\\
	&+(|a_1|-|a_2|)\sqrt{(\epsilon^A-\epsilon^B)^2+\gamma^2}\\
	\mathcal{C}(\varrho_d;H_{AB})&=|a_3|\sqrt{(\epsilon^A+\epsilon^B)^2+\gamma^2}\\
	&+|a_3|\sqrt{(\epsilon^A-\epsilon^B)^2+\gamma^2}.
	\end{split}		
	\end{equation*}
	\end{small}
	Note the fact that the reduced states of $\varrho$ are $\varrho_A=\varrho_B=\frac{1}{2}I_2$, which means that $\mathcal{C}(\varrho_A;H_A)=\mathcal{C}(\varrho_B;H_B)=0$. So in this case, the residual battery capacity $\bigtriangleup\mathcal{C}(\varrho)=\mathcal{C}(\varrho;H_{AB})$, the $RBC$ of incoherent part 
	\begin{small}
	\begin{equation*}
	RBC_{ic}=|a_3|\sqrt{(\epsilon^A+\epsilon^B)^2+\gamma^2}+|a_3|\sqrt{(\epsilon^A-\epsilon^B)^2+\gamma^2},
	\end{equation*}
	\end{small}
	and the $RBC$ of coherent part
	\begin{small}
	\begin{equation*}
	\begin{split}
	RBC_c&=(|a_1|+|a_2|-|a_3|)\sqrt{(\epsilon^A+\epsilon^B)^2+\gamma^2}\\
	     &+(|a_1|-|a_2|-|a_3|)\sqrt{(\epsilon^A-\epsilon^B)^2+\gamma^2}.
	\end{split}
	\end{equation*}
	\end{small}
	  From the expression of $RBC_c$, we can see that the $RBC$ of coherent part is directly proportional to $|a_1|$ and $|a_2|$. This is reasonable and intuitive, as an increase in $|a_1|$ and $|a_2|$ will enhance the coherence of $\varrho$, and the contribution of incoherent part $\mathcal{C}(\varrho_d;H_{AB})$ is related to $a_3$, which only appear on the diagonal of the density matrix.
	
	Now we consider optimizing the battery capacity distribution to achieve the capacity gain. One can use unitary matrix $U_{24}$ (the matrix obtained by exchanging the second row and the fourth row of the identity matrix) such that
	$$
	\Tilde{\varrho}=\frac{1}{4}\left(\begin{array}{cccc}
		1+a_3 & a_1-a_2 & 0 & 0 \\
		a_1-a_2 & 1+a_3 & 0 & 0 \\
		0 & 0 & 1-a_3 & a_1+a_2 \\
		0 & 0 & a_1+a_2 & 1-a_3
	\end{array}\right).
	$$
	Thus we have
	$$
	\Tilde{\varrho}_{A}=\frac{1}{2}\left(\begin{array}{cc}
		1+a_3 & 0 \\
		0 & 1-a_3
	\end{array}\right)
	$$
	and
	$$
	\Tilde{\varrho}_{B}=\frac{1}{2}\left(\begin{array}{cc}
		1 & a_1 \\
		a_1 & 1
	\end{array}\right),
	$$
	If we consider $a_1=0.5$, $a_2=0.3$ and $a_3=0.1$, then $\mathcal{C}(\Tilde{\varrho}_A;H_A)=0.2\epsilon^A<RBC_{ic}(\varrho)$, $\mathcal{C}(\Tilde{\varrho}_B;H_B)=\epsilon^B<RBC_c(\varrho)$. In summary, our unitary evolution transforms part of $RBC$ into the battery capacity of subsystem $A$ and subsystem $B$ to achieve capacity gain.
	
	\section{III. The distribution relation of battery capacity for n-qubit X state}
	
	Now, we extend the results of the two -qubit $X$ state to the n-qubit $X$ state.
	
	Consider the n-qubit $X$ state $\rho$ under the computational basis,
	$$\left(
	\begin{array}{ccccc}
		\rho_{11}& 0 &\cdots & 0 & \rho_{1,2^n}\\
		0 & \rho_{22} &\cdots & \rho_{2,2^n-1} & 0\\
		\vdots & \vdots & \ddots & \vdots & \vdots \\
		0 & \rho_{2^n-1,2} & \cdots & \rho_{2^n-1,2^n-1} & 0 \\
		\rho_{2^n,1} & 0 & \cdots & 0 & \rho_{2^n,2^n}
	\end{array} 
	\right ),$$
	where $\sum_{i=1}^{2^n}\rho_{ii}=1$ and $\rho_{ii}\geq0$. The eigenvalues of $\rho$ are
	\begin{equation*}
		\begin{split}
			&\lambda_0=\frac{1}{2}(\rho_{11}+\rho_{2^n,2^n}-\sqrt{(\rho_{11}-\rho_{2^n,2^n})^2+4|\rho_{1,2^n}|^2}),\\
			&\lambda_1=\frac{1}{2}(\rho_{11}+\rho_{2^n,2^n}+\sqrt{(\rho_{11}-\rho_{2^n,2^n})^2+4|\rho_{1,2^n}|^2}),\\
			&\cdots\\
			&\lambda_{2^n-2}=\frac{1}{2}(\rho_{2^{n-1},2^{n-1}}+\rho_{2^{n-1}+1,2^{n-1}+1}-\\
			&\sqrt{(\rho_{2^{n-1},2^{n-1}}-\rho_{2^{n-1}+1,2^{n-1}+1})^2+4|\rho_{2^{n-1},2^{n-1}+1}|^2}),\\
			&\lambda_{2^n-1}=\frac{1}{2}(\rho_{2^{n-1},2^{n-1}}+\rho_{2^{n-1}+1,2^{n-1}+1}+\\
			&\sqrt{(\rho_{2^{n-1},2^{n-1}}-\rho_{2^{n-1}+1,2^{n-1}+1})^2+4|\rho_{2^{n-1},2^{n-1}+1}|^2}).
		\end{split}
	\end{equation*}
	Herein, the entire system Hamiltonian is
	\begin{equation}\label{e14}
	\begin{split}
	H=&H_{A_1}\otimes I_2\otimes...\otimes I_2+I_2\otimes H_{A_2}\otimes...\otimes I_2\\
	&+...+I_2\otimes...\otimes I_2\otimes H_{A_n}+\gamma\sigma_1^{\otimes n}\\
	 =&\epsilon^{A_1}\sigma_3\otimes I_2\otimes...\otimes I_2+I_2\otimes\epsilon^{A_2}\sigma_3\otimes...\otimes I_2\\
	 &+...+I_2\otimes...\otimes I_2\otimes\epsilon^{A_n}\sigma_3+\gamma\sigma_1^{\otimes n},
	\end{split}
	\end{equation}
	where $\epsilon^{A_1}\geq\epsilon^{A_2}\geq...\geq\epsilon^{A_n}\geq0$, and $\gamma\geq0$ is the interaction parameter.

	In addition, the result in Lemma 2 still holds in the n-qubit system.
	\begin{lemma}
	For a n-qubit incoherent state $\rho_\text{ic}$ with diagonal elements $\rho_{11}, \rho_{22},..., \rho_{2^n,2^n}$, one have
	\begin{equation}\label{e15}
	\sum_{i=1}^{n}\mathcal{C}(\rho_\text{ic}^{A_i};H_{A_i})\leq\mathcal{C}(\rho_\text{ic};H).
	\end{equation}
	\end{lemma}
	Similar to the discussion in the two-qubit case, we can get the sequential rotation of  $\rho_{11}\geq\rho_{22}\geq... \geq\rho_{2^n,2^n}$ and a total of $2^n$ cases in the $n$-qubit incoherent state $\rho$ can make the equal sign of Eq.(\ref{e15}) hold when $\gamma=0$. 
	
	Moreover, we note the fact that for the $X$-state $\rho$, its decoherent state $\tau$ and $\rho$ have the same reduced density matrix. That is,
	\begin{equation}\label{e17}
	\rho_{A_i}=\tau_{A_i},\,i=1,2,...,n.
	\end{equation}
	Now, we give the main result of this section.
	\begin{theorem}
	Given the n-qubit $X$ state $\rho$, one has
	\begin{small}
	\begin{equation}\label{e18}
		\mathcal{C}(\rho_{A_1};H_{A_1})+\mathcal{C}(\rho_{A_2};H_{A_2})+...+\mathcal{C}(\rho_{A_n};H_{A_n})\leq\mathcal{C}(\rho;H).
	\end{equation}
	\end{small}
	If $\rho$ is an incoherent state satisfied its diagonal elements meet one of the $2^n$ orders above, and $\gamma=0$, then the equal sign holds.
	\end{theorem}
	\begin{proof}
	According to Lemma 1, we have
	\begin{equation}\label{e19}
	(\rho_{11},\rho_{22},...,\rho_{2^n,2^n})\prec(\lambda_0,\lambda_1,...,\lambda_{2^n-1}),
	\end{equation}
	where $\rho_{ii}$ and $\lambda_i$ are the diagonal elements and eigenvalues of $\rho$, respectively. Then one has
	\begin{small}
	\begin{equation*}
		\begin{split}
			\mathcal{C}(\rho;H)&\geq\mathcal{C}(\tau;H)\\
			&\geq\mathcal{C}(\tau_{A_1};H_{A_1})+\mathcal{C}(\tau_{A_2};H_{A_2})+...+\mathcal{C}(\tau_{A_n};H_{A_n})\\
			&=\mathcal{C}(\rho_{A_1};H_{A_1})+\mathcal{C}(\rho_{A_2};H_{A_2})+...+\mathcal{C}(\rho_{A_n};H_{A_n}).
		\end{split}
	\end{equation*}
	\end{small}
	The first inequality is due to the combination of Eq. (\ref{e19}) with Proposition 1, the second inequality is based on Lemma 3, and the last equation follows from Eq.(\ref{e17}). Furthermore, if $\rho$ is an incoherent state such that its diagonal elements satisfy the ordering conditions mentioned above and $\gamma=0$, then the inequalities above become equalities.
	\end{proof}
	 We can also define the RBC for an n-qubit $X$ state as
	\begin{equation}\label{e20}
		\bigtriangleup\mathcal{C}=\mathcal{C}(\rho;H)-\mathcal{C}(\rho_{A_1};H_{A_1})-...-\mathcal{C}(\rho_{A_n};H_{A_n}),
	\end{equation}
	define the $RBC$ of the incoherent part and $RBC$ of the coherent part as
	\begin{equation}\label{e21}
		\begin{split}
			&RBC_{ic}(\rho)=\mathcal{C}(\tau;H)-\sum_{i=1}^{n}\mathcal{C}(\rho_{A_i};H_{A_i}),\\
			&RBC_c(\rho)=\mathcal{C}(\rho;H)-\mathcal{C}(\tau;H),
		\end{split}
	\end{equation}
	respectively, where $\tau$ is the decoherent state of $\rho$.
	
	Similar to the proof of Theorem 2, we can achieve the capacity gain of the $n$-qubit $X$ state through a unitary matrix (possibly the product of a series of unitary matrices).
	\begin{theorem}
		For a given $n$ qubits $X$ state $\rho$, there is always a unitary evolution $U$ that enables the subsystems to achieve battery capacity gain.
	\end{theorem}
	
	This distribution relationship provides upper and lower bounds on the genuine battery capacity of each subsystem, and these bounds still hold for the $n$ qubits $X$ state.
	\begin{observation}
		The genuine battery capacity $\mathcal{C}_{A_i}$ of subsystem $A_i$ for n-qubit $X$ state $\rho$ satisfies
		\begin{equation}\label{e22}
			\mathcal{C}(\rho_{A_i};H_{A_i})\leq \mathcal{C}_{A_i}\leq\mathcal{C}(\rho_{A_i};H_{A_i})+\bigtriangleup\mathcal{C}, i=1,...,n.
		\end{equation}
	\end{observation}
	
Specifically, the monogamy relation (\ref{e18}) for the three-qubit case is
\begin{equation}\label{e23}
	\mathcal{C}(\rho_{A};H_{A})+\mathcal{C}(\rho_{B};H_{B})+\mathcal{C}(\rho_{C};H_{C})\leq\mathcal{C}(\rho;H).
\end{equation}
Such a result is reasonable. The battery capacity of the whole system includes not only the battery capacity of subsystems $A$, $B$, and $C$ but also the common battery capacity of the three subsystems. However, like the coherence information on the minor diagonal and some incoherence information on the diagonal, the RBC is lost due to the reduced density matrix. See Figure 3.
\begin{figure}[htbp]
	\centering
	\includegraphics[width=0.52\textwidth]{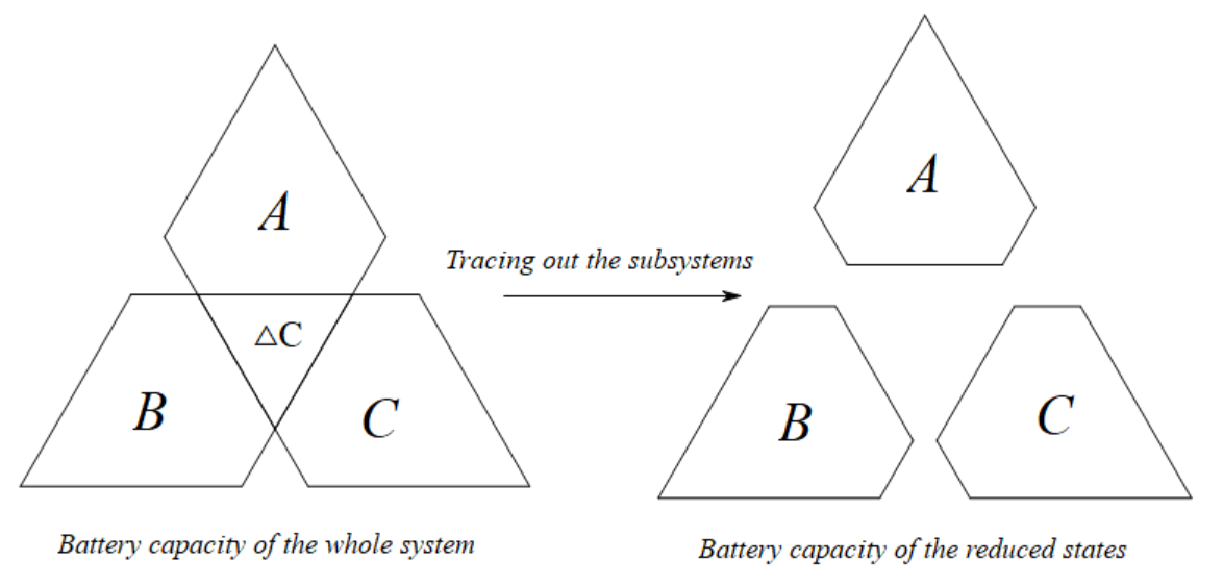}
	\vspace{-1em} \caption{The relationship between $\mathcal{C}(\rho_A;H_A)$, $\mathcal{C}(\rho_B;H_B)$, $\mathcal{C}(\rho_C;H_C)$ and $\bigtriangleup\mathcal{C}$.} \label{Fig.3}
\end{figure}
	
	Moreover, we can prove another capacity distribution relationship with respect to any three qubit $X$ state.
	\begin{theorem}
	Given a three-qubit $X$ state $\rho_{ABC}$, the following inequalities hold:
	\begin{small}
		\begin{equation}\label{e24}
		\begin{split}
		&\mathcal{C}(\rho_{ABC};H_{ABC})\geq\mathcal{C}(\rho_{AB};H_{AB})+\mathcal{C}(\rho_{C};H_{C})+RBC_c,\\
		&\mathcal{C}(\rho_{ABC};H_{ABC})\geq\mathcal{C}(\rho_{AC};H_{AC})+\mathcal{C}(\rho_{B};H_{B})+RBC_c,\\
		&\mathcal{C}(\rho_{ABC};H_{ABC})\geq\mathcal{C}(\rho_{BC};H_{BC})+\mathcal{C}(\rho_{A};H_{A})+RBC_c,\\
		\end{split}
	\end{equation}
	\end{small}
	where $RBC_c$ is the residual battery capacity of coherent part.
	\end{theorem}
	\begin{proof}
	We only prove the first inequality, and the proof approach for the remaining two is similar.
	
	Note that
	\begin{equation}\label{e25}
	\mathcal{C}(\rho_{ABC};H_{ABC})=\mathcal{C}(\tau_{ABC};H_{ABC})+RBC_c,
	\end{equation}
	here $\tau_{ABC}$ is the decoherent state of $\rho_{ABC}$. According to Eq. (\ref{e17}), we only need to demonstrate that
	\begin{small}
		\begin{equation}\label{e26}
				\mathcal{C}(\tau_{ABC};H_{ABC})\geq\mathcal{C}(\tau_{AB};H_{AB})+\mathcal{C}(\tau_{C};H_{C}).
		\end{equation}
	\end{small}
	Let the diagonal elements of $\tau_{ABC}$ be $\tau_{11}, \tau_{22},\dots,\tau_{88}$, then the reduced density matrices for subsystems $AB$ and $C$ are
	$$
	\tau_{AB}=\left(\begin{array}{cccc}
		\tau_{11}+\tau_{22} & 0 & 0 & 0 \\
		0 & \tau_{33}+\tau_{44} & 0 & 0 \\
		0 & 0 & \tau_{55}+\tau_{66} & 0 \\
		0 & 0 & 0 & \tau_{77}+\tau_{88}
	\end{array}\right),
	$$
	and
	$$
	\tau_{C}=\left(\begin{array}{cc}
		\tau_{11}+\tau_{33}+\tau_{55}+\tau_{77} & 0 \\
		0 & \tau_{22}+\tau_{44}+\tau_{66}+\tau_{88}
	\end{array}\right).
	$$
	We set the descending order of $\tau_{11},\tau_{22},...,\tau_{88}$ as $\omega_1\geq\omega_2\geq...\geq\omega_8$, and the descending order of the diagonal elements of $\tau_{AB}$ is $\mu_1\geq\mu_2\geq\mu_3\geq\mu_4$. Therefore, one can calculate that
	 \begin{footnotesize}
	 \begin{equation*}
	 	\begin{split}
	 		&\mathcal{C}(\tau_{ABC};H_{ABC})=2(\omega_1-\omega_8)\epsilon_1+2(\omega_2-\omega_7)\epsilon_2+2(\omega_3-\omega_6)\epsilon_3\\
	 		&+2(\omega_4-\omega_5)\epsilon_4\\
	 		&\mathcal{C}(\tau_{AB};H_{AB})=2(\mu_1+\mu_2-\mu_3-\mu_4)\epsilon^A\\
	 		&+2(\mu_1+\mu_3-\mu_2-\mu_4)\epsilon^B,\\
	 		&\mathcal{C}(\tau_C;H_C)=2|\tau_{11}+\tau_{33}+\tau_{55}+\tau_{77}-\tau_{22}-\tau_{44}-\tau_{66}-\tau_{88}|\epsilon^C,
	 	\end{split}
	 \end{equation*}
	 \end{footnotesize}
	 where
	 \begin{footnotesize}
	 \begin{equation*}
	 	\begin{split}
	 	&\epsilon_1=\sqrt{(\epsilon^A+\epsilon^B+\epsilon^C)^2+\gamma^2}, \epsilon_2=\sqrt{(\epsilon^A+\epsilon^B-\epsilon^C)^2+\gamma^2},\\
	 	&\epsilon_3=\sqrt{(\epsilon^A-\epsilon^B+\epsilon^C)^2+\gamma^2}, \epsilon_4=\sqrt{(\epsilon^A-\epsilon^B-\epsilon^C)^2+\gamma^2}.
	 	\end{split}
	 \end{equation*}
	 \end{footnotesize}
	 
	 Given a state $\tau_{ABC}$, $\mathcal{C}(\tau_{ABC};H_{ABC})$ is a constant value because the order of elements $\tau_{11},...\tau_{88}$ is determined. Thus we consider the case where $\mathcal{C}(\tau_{AB};H_{AB})+\mathcal{C}(\tau_C;H_C)$ takes the maximum value, which corresponds to
	 \begin{equation*}
	 \begin{split}
	 &\mu_1+\mu_2=\omega_1+\omega_2+\omega_3+\omega_4,\\
	 &\mu_3+\mu_4=\omega_5+\omega_6+\omega_7+\omega_8,\\
	 &\mu_1+\mu_3=\omega_1+\omega_2+\omega_5+\omega_6,\\
	 &\mu_2+\mu_4=\omega_3+\omega_4+\omega_7+\omega_8,
	 \end{split}
	 \end{equation*}
	 and 
	 \begin{small}
	 \begin{equation*}
	 	\mathcal{C}(\tau_C;H_C)=2(\omega_1+\omega_3+\omega_5+\omega_7-\omega_2-\omega_4-\omega_6-\omega_8)\epsilon^C.
	 \end{equation*}
	 \end{small}
	 
	 In this case,
	 \begin{small}
	 	\begin{equation*}
	 	\begin{split}
	 	\mathcal{C}(\tau_{ABC};H_{ABC})&=2(\omega_1-\omega_8)\epsilon_1+2(\omega_2-\omega_7)\epsilon_2+2(\omega_3-\omega_6)\epsilon_3\\
	 	&+2(\omega_4-\omega_5)\epsilon_4\\
	 	&\geq2(\omega_1+\omega_2+\omega_3+\omega_4-\omega_5-\omega_6-\omega_7-\omega_8)\epsilon^A\\
	 	&+2(\omega_1+\omega_2+\omega_5+\omega_6-\omega_3-\omega_4-\omega_7-\omega_8)\epsilon^B\\
	 	&+2(\omega_1+\omega_3+\omega_5+\omega_7-\omega_2-\omega_4-\omega_6-\omega_8)\epsilon^C\\
	 	&=\mathcal{C}(\tau_{AB};H_{AB})+\mathcal{C}(\tau_C;H_C).
	 	\end{split}
	 	\end{equation*}
	 \end{small}
	 So we have
	 \begin{small}
	 \begin{equation*}
	 	\begin{split}
	 		\mathcal{C}(\rho_{ABC};H_{ABC})&=\mathcal{C}(\tau_{ABC};H_{ABC})+RBC_c\\
	 		&\geq\mathcal{C}(\tau_{AB};H_{AB})+\mathcal{C}(\tau_{C};H_{C})+RBC_c\\
	 		&=\mathcal{C}(\rho_{AB};H_{AB})+\mathcal{C}(\rho_{C};H_{C})+RBC_c.
	 	\end{split}
	 \end{equation*}
	 \end{small}
	\end{proof}
	
	\mathbi{Example 2.}
	In addition to the three distribution relationships in Theorem 3, the following capacity relationships may also be intuitively correct:
	\begin{small}
		\begin{equation*}
			\begin{split}
				&\mathcal{C}(\rho_{ABC};H_{ABC})\geq\mathcal{C}(\rho_{AB};H_{AB})+\mathcal{C}(\rho_{AC};H_{AC})-\mathcal{C}(\rho_A;H_A),\\
				&\mathcal{C}(\rho_{ABC};H_{ABC})\geq\mathcal{C}(\rho_{AB};H_{AB})+\mathcal{C}(\rho_{BC};H_{BC})-\mathcal{C}(\rho_B;H_B),\\
				&\mathcal{C}(\rho_{ABC};H_{ABC})\geq\mathcal{C}(\rho_{AC};H_{AC})+\mathcal{C}(\rho_{BC};H_{BC})-\mathcal{C}(\rho_C;H_C).
			\end{split}
		\end{equation*}
	\end{small}
Unfortunately, we can find corresponding states that violate these three inequalities. For simplicity, one consider the case of $\gamma=0$ as follows.

For the first inequality, we consider
$$\rho^1=\frac{1}{36}\left(
\begin{array}{cccccccc}
	8 & 0 & 0 & 0 & 0 & 0 & 0 & 0\\
	0 & 7 & 0 & 0 & 0 & 0 & 0 & 0\\
	0 & 0 & 2 & 0 & 0 & 0 & 0 & 0\\
	0 & 0 & 0 & 1 & 0 & 0 & 0 & 0\\
	0 & 0 & 0 & 0 & 4 & 0 & 0 & 0\\
	0 & 0 & 0 & 0 & 0 & 3 & 0 & 0\\
	0 & 0 & 0 & 0 & 0 & 0 & 6 & 0\\
	0 & 0 & 0 & 0 & 0 & 0 & 0 & 5\\
\end{array} 
\right ).$$
By calculation, we have
\begin{equation*}
\begin{split}
&\mathcal{C}(\rho^1;H_{ABC})=\frac{1}{36}(32\epsilon^A+16\epsilon^B+8\epsilon^C),\\
&\mathcal{C}(\rho^1_{AB};H_{AB})=\frac{1}{36}(32\epsilon^A+16\epsilon^B),\\
&\mathcal{C}(\rho^1_{AC};H_{AC})=\frac{8\epsilon^A}{36}, \mathcal{C}(\rho_A^1;H_A)=0.\\
\end{split}
\end{equation*}
So, we have
\begin{small}
	\begin{equation*}
			\mathcal{C}(\rho^1;H_{ABC})\leq\mathcal{C}(\rho_{AB}^1;H_{AB})+\mathcal{C}(\rho_{AC}^1;H_{AC})-\mathcal{C}(\rho_A^1;H_A),
	\end{equation*}
\end{small}
which violates the first inequality.

For the second inequality, consider the state
$$\rho^2=\frac{1}{36}\left(
\begin{array}{cccccccc}
	8 & 0 & 0 & 0 & 0 & 0 & 0 & 0\\
	0 & 7 & 0 & 0 & 0 & 0 & 0 & 0\\
	0 & 0 & 4 & 0 & 0 & 0 & 0 & 0\\
	0 & 0 & 0 & 3 & 0 & 0 & 0 & 0\\
	0 & 0 & 0 & 0 & 2 & 0 & 0 & 0\\
	0 & 0 & 0 & 0 & 0 & 1 & 0 & 0\\
	0 & 0 & 0 & 0 & 0 & 0 & 6 & 0\\
	0 & 0 & 0 & 0 & 0 & 0 & 0 & 5\\
\end{array} 
\right ).$$
The calculation shows that
\begin{equation*}
	\begin{split}
		&\mathcal{C}(\rho^2;H_{ABC})=\frac{1}{36}(32\epsilon^A+16\epsilon^B+8\epsilon^C),\\
		&\mathcal{C}(\rho^2_{AB};H_{AB})=\frac{1}{36}(32\epsilon^A+16\epsilon^B),\\
		&\mathcal{C}(\rho^2_{BC};H_{BC})=\frac{8\epsilon^B}{36}, \mathcal{C}(\rho_B^2;H_B)=0.\\
	\end{split}
\end{equation*}
Therefore, one has
\begin{small}
	\begin{equation*}
		\mathcal{C}(\rho^2;H_{ABC})\leq\mathcal{C}(\rho_{AB}^2;H_{AB})+\mathcal{C}(\rho_{BC}^2;H_{BC})-\mathcal{C}(\rho_B^2;H_B),
	\end{equation*}
\end{small}
which violates the second inequality. 

For the last inequality, the considered state is
$$\rho^3=\frac{1}{36}\left(
\begin{array}{cccccccc}
	8 & 0 & 0 & 0 & 0 & 0 & 0 & 0\\
	0 & 4 & 0 & 0 & 0 & 0 & 0 & 0\\
	0 & 0 & 7 & 0 & 0 & 0 & 0 & 0\\
	0 & 0 & 0 & 3 & 0 & 0 & 0 & 0\\
	0 & 0 & 0 & 0 & 6 & 0 & 0 & 0\\
	0 & 0 & 0 & 0 & 0 & 2 & 0 & 0\\
	0 & 0 & 0 & 0 & 0 & 0 & 5 & 0\\
	0 & 0 & 0 & 0 & 0 & 0 & 0 & 1\\
\end{array} 
\right ).$$
It can be verified that
\begin{equation*}
	\begin{split}
		&\mathcal{C}(\rho^3;H_{ABC})=\frac{1}{36}(32\epsilon^A+16\epsilon^B+8\epsilon^C),\\
		&\mathcal{C}(\rho^3_{AC};H_{AC})=\frac{1}{36}(32\epsilon^A+16\epsilon^C),\\
		&\mathcal{C}(\rho^3_{BC};H_{BC})=\frac{1}{36}(32\epsilon^B+8\epsilon^C), \mathcal{C}(\rho_C^3;H_C)=\frac{32\epsilon^C}{36}.\\
	\end{split}
\end{equation*}
Thus we have
\begin{small}
	\begin{equation*}
		\mathcal{C}(\rho^3;H_{ABC})\leq\mathcal{C}(\rho_{AC}^3;H_{AC})+\mathcal{C}(\rho_{BC}^3;H_{BC})-\mathcal{C}(\rho_C^3;H_C),
	\end{equation*}
\end{small}
which violates the last inequality. It is easy to see that $\mathcal{C}(\rho_{ABC};H_{ABC})$ is a continuous functional of $\gamma$. Therefore, if these inequalities are violated when $\gamma=0$, then there must be a positive real number $\gamma_c$ so that when $\gamma\in[0,\gamma_c)$, these three inequalities are violated.

\mathbi{Example 3.} Consider the $n$-qubit GHZ state $|GHZ\rangle=\frac{1}{\sqrt{2}}\sum_{i=0}^{1}|i\rangle^{\otimes n}$ affected by the white noise,
\begin{equation*}
\rho_\beta=\beta|GHZ\rangle\langle GHZ|+\frac{1-\beta}{2^n}\mathbbm{1},
\end{equation*}
where $0\leq\beta\leq1$.

Due to its central role in many quantum tasks \cite{kchk,lsbl,gjmg}, we believe it also has potential value in the research of quantum batteries. The eigenvalues of $\rho_\beta$ are
\begin{equation*}
\begin{split}
&\lambda_i=\frac{1-\beta}{2^n},\,i=0,...,2^n-2,\\
&\lambda_{2^n-1}=\frac{1-\beta}{2^n}+\beta,
\end{split}
\end{equation*}
and the reduced density matrices are $\rho_\beta^{A_1}=\rho_\beta^{A_2}=...=\rho_\beta^{A_n}=\frac{1}{2}I_2$. For the convenience, in all the numerical calculations we will consider $\epsilon^{A_1}=0.5$, $\epsilon^{A_2}=0.3$ and $\epsilon^{A_i}=0.1\,(i=3,...,n)$. According to Eq. (\ref{e2}), one have
\begin{small}
\begin{equation*}
	\begin{split}
	\mathcal{C}(\rho_\beta;H)&=2\beta\sqrt{(0.6+0.1n)^2+\gamma^2},\\
	\mathcal{C}(\tau_\beta;H)&=L_D=(\sqrt{(0.6+0.1n)^2+\gamma^2}+\sqrt{(0.4+0.1n)^2+\gamma^2})\beta,
	\end{split}	
\end{equation*}
\end{small}
and $\mathcal{C}(\rho_\beta^{A_1};H_{A_1})=\mathcal{C}(\rho_\beta^{A_2};H_{A_2})=...=\mathcal{C}(\rho_\beta^{A_n};H_{A_n})=0$.
Therefore, we have
\begin{small}
\begin{equation*}
	\begin{split}
		\bigtriangleup\mathcal{C}(\rho_\beta)&=\mathcal{C}(\rho_\beta;H)-\sum_{i=1}^{n}\mathcal{C}(\rho_\beta^{A_i};H_{A_i})\\
		&=2\beta\sqrt{(0.6+0.1n)^2+\gamma^2},\\
		RBC_{ic}&=\mathcal{C}(\tau_\beta;H)-\sum_{i=1}^{n}\mathcal{C}(\rho_\beta^{A_i};H_{A_i})\\
		&=(\sqrt{(0.6+0.1n)^2+\gamma^2}+\sqrt{(0.4+0.1n)^2+\gamma^2})\beta,\\
		RBC_c&=\mathcal{C}(\rho_\beta;H)-\mathcal{C}(\tau_\beta;H)\\
		&=(\sqrt{(0.6+0.1n)^2+\gamma^2}-\sqrt{(0.4+0.1n)^2+\gamma^2})\beta.
	\end{split}	
\end{equation*}
\end{small}
From the expressions for $RBC_{ic}$ and $RBC_c$, it can be seen that as $\beta$ approaches 1, both $RBC_{ic}$ and $RBC_c$ are increasing, which is not difficult to explain: the coherence of $\rho_\beta$ has been improving during this process, leading to an increase in $RBC_c$. And the increase of $\beta$ also improves the two largest eigenvalues of $\tau_\beta$, so $RBC_{ic}$ is also increasing. In particular, when $n=3$ and $\gamma=0.5$, the trend of $RBC$ and $RBC_{ic}$ is shown in Figure 4.
\begin{figure}[htbp]
	\centering
	\includegraphics[width=0.5\textwidth]{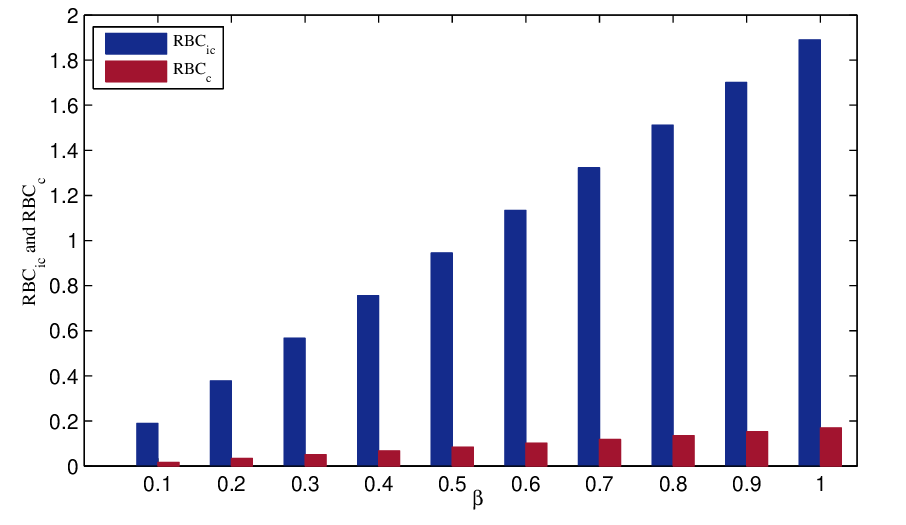}
	\vspace{-1em} \caption{When $\gamma=0.5$, the trend of $RBC_{ic}$ and $RBC_c$ for different values of $\beta$ as shown in figure.} \label{Fig.4}
\end{figure}

Now we can use unitary matrix $U_{2,2^n}$ (the matrix obtained by exchanging the second row and the $2^n$th row of the identity matrix) to achieve the battery capacity gain. Then the reduced states of $\Tilde{\rho}_\beta$ are
$$
\Tilde{\rho}_\beta^{A_i}=\frac{1}{2}\left(\begin{array}{cc}
	1+\beta & 0 \\
	0 & 1-\beta
\end{array}\right)
$$
for $i=1,...,n-1$, and
$$
\Tilde{\rho}_\beta^{A_n}=\frac{1}{2}\left(\begin{array}{cc}
	1 & \beta \\
	\beta & 1
\end{array}\right).
$$
Therefore, we have
\begin{small}
\begin{equation*}
	\begin{split}
		\sum_{i=1}^{n}\mathcal{C}(\Tilde{\rho}_\beta^{A_i};H_{A_i})&=\sum_{i=1}^{n-1}\mathcal{C}(\Tilde{\rho}_\beta^{A_i};H_{A_i})+\mathcal{C}(\Tilde{\rho}_\beta^{A_n};H_{A_n})\\
		&=(1.2+0.2n)\beta.
	\end{split}
\end{equation*} 
\end{small}
The above equation means that our unitary operation transforms part of residual battery capacity into subsystems. We use capacity gain to represent the difference between $\sum_{i=1}^{n}\mathcal{C}(\rho_\beta^{A_i};H_{A_i})$ and $\sum_{i=1}^{n}\mathcal{C}(\Tilde{\rho}_\beta^{A_i};H_{A_i})$. In order to observe the capacity transfer efficiency, we calculate the ratio of capacity gain to the residual battery capacity of $\rho_\beta$ for some values of the interaction parameter $\gamma$,
\begin{small}
\begin{equation*}
	\begin{split}
		\text{ratio}&=\frac{\sum_{i=1}^{n}\mathcal{C}(\Tilde{\rho}_\beta^{A_i};H_{A_i})-\sum_{i=1}^{n}\mathcal{C}(\rho_\beta^{A_i};H_{A_i})}{\bigtriangleup\mathcal{C(\rho_\beta)}}\\
		&=\frac{1.2+0.2n}{2\sqrt{(0.6+0.1n)^2+\gamma^2}}.
	\end{split}
\end{equation*}
\end{small}
When $\gamma=0$, the capacity transfer efficiency is $1$, that is to say, we use unitary evolution to transfer all the residual battery capacity to the subsystems. However, as the interaction parameter $\gamma$ gradually increase, the capacity transfer efficiency will decrease. This is intuitive because the interaction between subsystems will dissipate part of the battery capacity.

Furthermore, the purpose of using unitary evolution is to increase the proportion of subsystems capacity by compressing the proportion of $RBC$. We calculate the proportion of final state's $RBC$ in the entire battery capacity $\mathcal{C}(\Tilde{\rho}_\beta;H)$ for some values of parameters $\gamma$ and $n$, as shown in Figure 5.
\begin{figure}[htbp]
	\centering
	\includegraphics[width=0.45\textwidth]{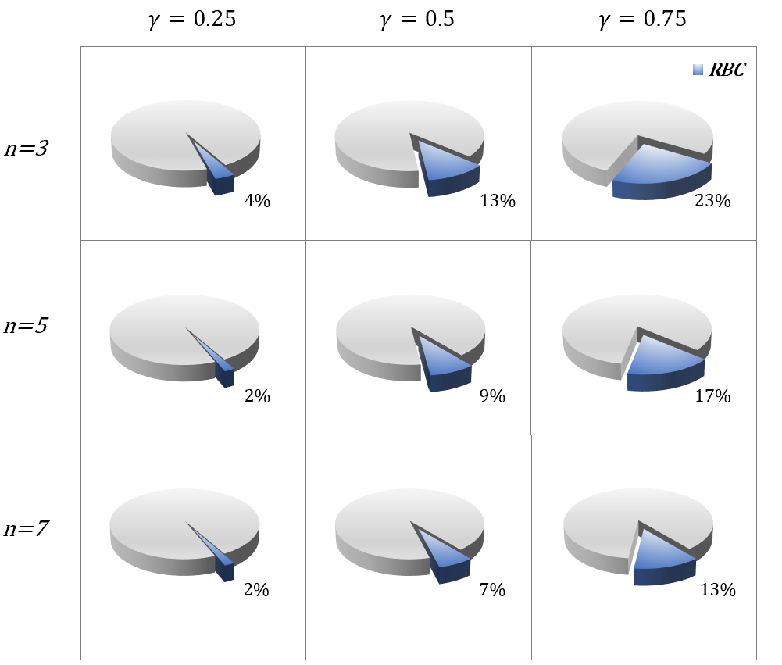}
	\vspace{-1em} \caption{The proportion of $RBC$ in the entire battery capacity for some values of $\gamma$ and $n$ as shown in figure.} \label{Fig.5}
\end{figure}

\section{IV. Conclusions and discussions}
We have investigated the monogamy relationships in quantum battery capacity for the first time. First, we proved the monogamy relation of battery capacity for any two-qubit $X$ states and provided the lower bound of capacity for any battery state. We also defined the concept of residual battery capacity ($RBC$), showing that it can be divided into coherent and incoherent parts. In addition, we demonstrated that global unitary evolution can be used to increase the battery capacity of subsystems at the cost of compressing the residual battery capacity. Furthermore, we observed that the distributive relation for battery capacity can be extended to general $n$ qubits $X$-states and any $n$ qubits $X$-state's battery capacity distribution can be optimized to achieve capacity gain through an appropriate global unitary evolution. Specifically, for three qubits $X$ states, we proposed stronger monogamy relations for battery capacity and provided counterexamples for a set of intuitively correct monogamy relationships.

In this paper, we mainly consider the $n$ qubits $X$-state. future work could explore more general quantum states, raising the interesting question of whether a monogamy relation exists for general quantum battery states or if some other, possibly stronger or weaker, relations can be formulated.

\bigskip
{\bf Acknowledgments:} ~This work is supported by the National Natural Science Foundation of China (NSFC) under Grant Nos. 12204137, 12075159, and 12171044; the specific research fund of the Innovation Platform for Academicians of Hainan Province.

\end{document}